\documentclass{eptcs}
\providecommand{\event}{CL\&C 2010}
\usepackage{amsmath,amssymb,amsthm}
\newtheorem{definition}{Definition}
\newtheorem{theorem}[definition]{Theorem}
\newtheorem{proposition}[definition]{Proposition}
\newtheorem{remark}[definition]{Remark}
\newtheorem{lemma}[definition]{Lemma}
\newtheorem{corollary}[definition]{Corollary}
\newtheorem{question}[definition]{Question}

\newcommand{\rrandnotiz}[1]{}
\newcommand{\randnotiz}[1]{}
\newcommand{\PC}{\textit{\textsf{PC}\/}}
\newcommand{\PRN}{\textit{\textsf{PRN}\/}}
\newcommand{\PPR}{\textit{\textsf{PPR}\/}}
\newcommand{\BRN}{\textit{\textsf{BRN}\/}}
\newcommand{\BPR}{\textit{\textsf{BPR}\/}}
\newcommand{\MBPR}{\textit{\textsf{MBPR}\/}}
\newcommand{\MPPR}{\textit{\textsf{MPPR}\/}}
\newcommand{\PR}{\textit{\textsf{PR}\/}}

\newcommand{\CC}{\textit{\textsf{C}}}

\newcommand{\pleft}[1]{\textsl{left}(#1)}
\newcommand{\pright}[1]{\textsl{right}(#1)}
\newcommand{\pnumber}[1]{\textsl{no}(#1)}
\newcommand{\pselect}[1]{\textsl{sel}(#1)}

\newcommand{\aph}{\textsf{\textup{APH}}}
\newcommand{\seq}[2]{#1 \Rightarrow #2}
\newcommand{\aphp}{\aph^+}
\newcommand{\rr}[2]{#1\ \triangleright\ #2}
\newcommand{\We}{\textsf{W}}
\newcommand{\WW}{\mathbb{W}}

\newcommand{\MM}{\mathcal{M}}
\newcommand{\MMle}{\MM(\lambda\eta)}

\newcommand{\sz}[1]{\textsf{s}_{\textsf{0}}\,#1}
\newcommand{\so}[1]{\textsf{s}_{\textsf{1}}\,#1}
\newcommand{\sll}[1]{\textsf{s}_{\ell}\,#1}
\newcommand{\pll}[1]{\textsf{p}_{\ell}\,#1}
\newcommand{\pw}[1]{\textsf{p}_{\We}\,#1}
\newcommand{\cw}{\textsf{c}_\We}
\newcommand{\csubs}{\textsf{c}_{\subseteq}}
\newcommand{\cmono}{\textsf{c}_{\succeq}}
\newcommand{\lw}{\textsf{l}_\We}
\newcommand{\ka}{\textsf{k}}
\newcommand{\ess}{\textsf{s}}

\newcommand{\sbwiw}{(\Sigma_\We^\textbf{b}\text{-}\textsf{I}_\We)}
\newcommand{\sbwmpi}{(\Sigma_\We^\textbf{b}\text{-}\textsf{MPI})}

\newcommand{\II}{\mathcal{I}}
\newcommand{\BB}{\mathcal{B}}
\newcommand{\Bb}{\textsf{B}}

\newcommand{\izo}{\ensuremath{i \in \{0,1\}}}

\newcommand{\FPtime}{\textsc{\textup{FPtime}}}
\newcommand{\FPspace}{\textsc{\textup{FPspace}}}
\newcommand{\FPH}{\textsf{\textup{FPH}}}

\title{An applicative theory for $\FPH$}
\author{Reinhard Kahle
\institute{CENTRIA and DM, FCT, Universidade Nova de Lisboa, P-2829-516
  Caparica, Portugal}
\email{kahle@mat.uc.pt}
\and
Isabel Oitavem
\institute{CMAF, Universidade de Lisboa and DM, FCT, Universidade Nova de
Lisboa, P-2829-516 Caparica, Portugal}
\email{oitavem@fct.unl.pt}}

\def\titlerunning{An applicative theory for $\FPH$}
\def\authorrunning{Reinhard Kahle and Isabel Oitavem}

\begin{document}

\maketitle
{\renewcommand{\thefootnote}{}\footnotetext{Work
    partially supported by the ESF research project \emph{Dialogical
      Foundations of Semantics} within the ESF Eurocores program
    \emph{LogICCC}, LogICCC/0001/2007 (funded by the Portuguese
    Science Foundation, FCT). The second author was also supported by
    the project \emph{Functional interpretations of arithmetic and
      analysis}, PTDC/MAT/104716/2008 from FCT.}}

\begin{abstract}
In this paper we introduce an applicative theory which characterizes
the polynomial hierarchy of time.
%\keywords{Computational Complexity, Applicative Theories, Polynomial Hierarchy}
\end{abstract} 

\section{Introduction}\randnotiz{First paragraph not very informative IV.1}

In this paper we define an applicative theory whose provably total
functions are those which belong to the polynomial hierarchy of time. 

Considering theories which characterize classes of computational
complexity, there are three different approaches: in one, the
functions which can be defined within the theory are ``automatically''
within a certain complexity class. In such an account, the syntax has
to be restricted to guarantee that one stays in the appropriate class.
This results, in general, in the problem that certain definitions of
functions do not work any longer, even if the function is in the
complexity class under consideration. In a second account, the
underlying logic is restricted.\footnote{As an example for this
  approach we may cite \cite{Schw06}.} In the third account, one does
not restrict the syntax, allowing, in general, to write down
``function terms'' for arbitrary (partial recursive) functions, nor
the logic, but only for those function terms which belong to the
complexity class under consideration, one can \emph{prove} that they
have a certain characteristic property, usually, the property that
they are ``provably total'' (see Definition \ref{provtotal} below).
While the function terms, according to the underlying syntactical
framework, may have a straightforward computational character, i.e.,
as $\lambda$ terms, the logic which is used to prove the
characteristic property may well be \emph{classical}.

Here, we follow the third account, using \emph{applicative theories}
as underlying framework. 

Applicative theories are the first-order part Feferman's system of
explicit mathematics \cite{Fef75,Fef79}. They provide a very handy
framework to formalize theories of different strength, including to
characterize classes of computational complexity. A first
characterization of polynomial time operations in applicative theories
was given by Strahm in \cite{Str97}. A uniform approach to varies
complexity classes, including \textsc{FPtime}, \textsc{FPspace},
\textsc{FPtime-FLinspace}, and \textsc{FLinspace} was given by the same
author in his Habilitationsschrift, published in \cite{Str03}. These
characterizations are based on bounded schemes in the vein of Cobham
\cite{Cob65} (see also \cite{Clo99}). Cantini \cite{Can02} gave, at
the same time, a characterization of \textsc{FPtime} in an applicative
framework following the approach of Bellantoni and Cook \cite{BC92}
which separates the input positions of functions in normal and safe. 

On the base of a characterization of the functions in the Polynomial
Hierarchy which uses a monotonicity condition, given in
\cite{BALO1x}, we present here an applicative theory for
\FPH. Given a function algebra, the main objective of defining
a corresponding theory is, of course, to introduce an adequate
induction scheme which allows to prove properties for the functions
under consideration. 
%The contribution of the present paper is, first
%of all, to define the induction scheme which takes care of the
%monotonicity conditions. 
In section \ref{zwei} we rewrite the input-sorted characterization of
$\FPH$ given in \cite{BALO1x} as a non-sorted
characterization, in Cobham style, by introducing bounds in the
recursion schemes. The next sections are concerned with the main goal
of this paper: to define an induction scheme which takes care of the
monotonicity condition.  While the proof of the lower bound follows
from a (more or less) straightforward embedding of the function
algebra %given in \cite{BALO1x}
described in section \ref{zwei}, the upper bound is carried out by an
adaptation of the proof(s) given by Strahm in \cite{Str03}.

Note, that Strahm also treats the polynomial hierarchy in
\cite{Str03}, but in a quite different way which involves a special
type two functional. 

\paragraph{Notation.} We use $\WW$ to denote the word algebra
generated by $\epsilon$ (source), and $S_0$ and $S_1$
(successors). $\WW$ is usually interpreted over the set of binary
words $\{0,1\}^*$. Given $x,y \in \WW$, $|x|$ is the length of $x$ and
$x|_y$ denotes the word corresponding to the first $|y|$ bits of $x$. $x'$
denotes the numeric successor of $x$, and it defined according to the
equations $\epsilon' = S_0(\epsilon)$, $(S_0(x))' = S_1(x)$ and
$(S_1(x))' = S_0(x')$. The letters $x,y,z,w,\dots$ denote usually
variables, while $f,g,h,s,r,\dots$ denote function symbols. $\vec x$
and $\vec f$ denote, respectively, a sequence of variables and
functions of the appropriate arity.

\section{Function algebras for $\FPH$}\label{zwei}

In this section we work with two function algebras. One formulated in
a non-sorted context, and the other formulated in a two-input-sorted
context following notation introduced by Bellantoni and Cook in
\cite{BC92}. In the sorted context, function arguments have two sorts,
\emph{normal} and \emph{safe}. We write them by this order, separated
by a semicolon: $f(\vec x;\vec y)$.

$\textsc{PH}$, the polynomial hierarchy of time, is usually defined as
$\bigcup_i \Sigma_i$ or $\bigcup_i \Delta_i$ with $\Sigma_0 = \Delta_0
= \textsc{P}$ and, for $i \ge 0$, $\Sigma_{i+1} =
\textsc{NP}(\Sigma_i)$ and $\Delta_{i+1} =
\textsc{P}(\Sigma_i)$. The corresponding function classes are $\Box_i
= \textsc{FPtime}(\Delta_i) = \textsc{FPtime}(\Sigma_{i-1})$, for $i
\ge 1$, and $\FPH = \bigcup_i \Box_i =
\textsc{FPtime}(\textsc{PH})$. 

Consider the following partial order over $\WW$, using $\leq$ as the
natural one on $\{0,1\}$. 

\begin{definition}
For $w,v \in \WW$, we write $w \preceq v$ if $|w| < |v|$, or $|w| =
|v|$ and $\forall i. w_i \leq v_i$. We write $w \prec v$ if $w \preceq
v$ but $w \not= v$. 
\end{definition}

\begin{definition}
\begin{enumerate}
\item A function $h$ is called \emph{monotone} if, for all $z \in
  \WW$, $z\preceq h(\vec{x},z)$.
\item A two-sorted function $h$, with at least one safe argument, is
  called \emph{monotone} if, for all $z \in \WW$, $z\preceq
  h(\vec{x};\vec{y},z)$.
\end{enumerate}
\end{definition}

\begin{definition}
\begin{enumerate}
\item 
Given a function $h$, its \emph{monotone section} is the function 
$$h^m(\vec x, z) = 
\begin{cases}
h(\vec x,z) & \text{if $z \preceq h(\vec x,z)$},\\
z & \text{otherwise}.
  \end{cases}$$
\item Given a two-sorted function $h$, with at least one safe
  argument, its \emph{monotone section} is the function 
$$h^m(\vec x;\vec{y}, z) = 
\begin{cases}
h(\vec x;\vec{y},z) & \text{if $z \preceq h(\vec x;\vec{y},z)$},\\
z & \text{otherwise}.
  \end{cases}$$
\end{enumerate}
\end{definition}  

Clearly, monotone sections are always monotone functions.

\subsection{Predicative approach}

Consider the class $[\BB; \PC, \PRN, \PPR]$ of two-input-sorted
functions.

$\BB$ is the set of basic functions defined as follows:
\begin{enumerate}
\item %Source:
$\epsilon$ (a zero-ary function);
\item %Projections:
$\pi_i^{k,n}(x_1,\dots,x_k;x_{k+1},\dots,x_{k+n}) = x_i$, for each $1
\le i \le k+n$;
\item %Normal binary successors:
$S_i(x;) = xi$, \izo;
\item %Bounded safe binary successors:
$S_i(z;x) = \begin{cases}
xi & \text{if $|x|<|z|$},\\
x & \text{otherwise},
\end{cases} \izo;$
\item %Binary predecessor:
$P(;\epsilon) = \epsilon$, $P(;xi) = x$, \izo;
\item %Numerical predecessor:
$p(;\epsilon) = \epsilon$, $p(;x') = x$;
\item %Conditional:
$Q(;\epsilon,y,z_0,z_1) = y$, $Q(;xi,y,z_0,z_1) = z_i$, \izo;
\item %Tally product:
  $\times(x,y;) = 1^{|x|\times|y|}$.
\end{enumerate}

$\PC$, $\PRN$ and $\PPR$ are the following operators:
\begin{itemize}
\item Predicative composition: Given $g, \vec r, \vec s$, their
  predicative composition $f = \PC(g,\vec r,\vec s)$ is defined by 
$$f(\vec x;\vec y) = g(\vec r(\vec x;);\vec s(\vec x; \vec y)).$$
\item Predicative recursion on notation: Given $g, h_0, h_1$, the
  predicative recursion on notation scheme defines a function
  $f=\PRN(g,h_0,h_1)$ by
\begin{align*}
f(\epsilon,\vec x;\vec y) & = g(\vec x;\vec y), \\
f(zi,\vec x;\vec y) & = h_i(z,\vec x;\vec y,f(z,\vec x; \vec z)), &\izo
\end{align*}
\item Predicative primitive recursion: Given $g$ and $h$, the
  predicative primitive recursion scheme defines a function $f =
  \PPR(g,h)$ by 
\begin{align*}
f(\epsilon,\vec x;\vec y) & = g(\vec x;\vec y), \\
f(z',\vec x;\vec y) & = h(z,\vec x;\vec y,f(z,\vec x; \vec z)).
\end{align*}
\end{itemize}

\begin{proposition}[\cite{BC92} and \cite{Oit97}]
\begin{itemize}
\item $[\BB; \PC, \PRN] = \FPtime$,
\item $[\BB;\PC, \PRN, \PPR] = \FPspace$.
\end{itemize}
\end{proposition}

\begin{definition}
Given $g$ and $h$, the \emph{predicative monotone primitive recursion
  scheme} $\MPPR$ is defined by $\MPPR(g,h) = \PPR(g,h^m)$. 
\end{definition}

\begin{proposition}[\cite{BALO1x}]\label{p6}
$[\BB;\PC, \PRN, \MPPR] = \FPH.$
\end{proposition}

\begin{remark}\label{r7} For all $f \in [\BB;\PC, \PRN, \PPR]$:
\begin{enumerate}
\item there exists a $F \in [\BB;\PC, \PRN, \PPR]$ such that 
$\forall \vec x, \vec y. F(\vec x,\vec y;) = f(\vec x; \vec y)$;
\item there exists a polynomial $q_f$ such that 
$\forall \vec x, \vec y. |f(\vec x;\vec y)| \le \max\{q_f(|\vec x|),
\max_i|y_i|\}$.
\end{enumerate}
This remark holds also if $[\BB;\PC, \PRN, \PPR]$ is replaced by
$[\BB;\PC, \PRN, \MPPR]$. 
\end{remark}
See \cite{Oit97} for details.

\subsection{Bounded approach}

Consider the class $[\II; \CC, \BRN, \BPR]$ where:
\begin{itemize}
\item $\II$ is the set of initial functions:
\begin{enumerate}
\item $\epsilon$,
\item $S_i(x) = xi$, \quad\izo,
\item $\pi^n_j(x_1,\dots,x_n) = x_j$, $1 \le j \le n$,
\item $Q(\epsilon,y,z_0,z_1) = y$, 
      $Q(xi,y,z_0,z_1) = z_i$, \quad\izo,
\item $\times(x,y) = 1^{|x|\times|y|}$.
\end{enumerate}
\item $C$, $\BRN$ and $\BPR$ are the following operators:
\begin{itemize}
\item Composition: Given $g$ and $\vec h$, their composition $f =
  C(g,\vec h)$ is given by $f(\vec x) = g(\vec h(\vec x))$,
\item Bounded recursion on notation: Given $g$, $h_0$, $h_1$, and $t$,
  the bounded recursion on notation $f = \BRN(g,h_0,h_1,t)$ is given by:
\begin{eqnarray*}
f(\epsilon,\vec x)  & = & g(\vec x) \\
f(yi,\vec x) & = & h_i(y,\vec x, f(y,\vec x)) |_{t(y,\vec x)}, \quad\izo
\end{eqnarray*}
\item Bounded primitive recursion:\ Given $g$, $h$, and $t$, the
  bounded primitive recursion $f = \linebreak\BPR(g,h,t)$ is given by
\begin{eqnarray*}
f(\epsilon,\vec x)  & = & g(\vec x) \\
f(y',\vec x) & = & h(y,\vec x, f(y,\vec x)) |_{t(y,\vec x)}
\end{eqnarray*}
\end{itemize}
%$t$ is a \emph{bounding function}, i.e., a
%function of the smallest class of functions containing the projection
%functions and the concatenation and binary product functions and which
%is closed under composition and assignement of values to variables.
\end{itemize}

\begin{proposition}
\begin{itemize}
\item $[\II; \CC, \BRN] = \FPtime$,
\item $[\II; \CC, \BRN, \BPR] = \FPspace$.
\end{itemize}
\end{proposition}

These are well-known results, essentially due to Cobham \cite{Cob65}
and Thompson \cite{Tho71},
here formulated over $\WW$. See \cite{Oit97} 
or \cite{Oit01} for a reference.

$\PR$ is the usual operator for primitive recursion, i.e., $f =
\PR(g,h)$ means that $f$ is defined by primitive recursion, with $g$ as
base function and $h$ as step function. 

\begin{definition}
Given $g,h,t$, the \emph{monotone bounded primitive recursion scheme}
is defined by
%$$MBPR(g,h,t) = BPR(g,h^m,t).$$
%$$MBPR(g,h,t) = BPR(g,(h|_t)^m).$$
$$\MBPR(g,h,t) = \PR(g,(h|_t)^m).$$
\end{definition}

\begin{remark}
Given a function $t(y,\vec x)$ in $[\II; \CC, \BRN, \MBPR]$, we may define
within the same class a function $t^+$, which is non-decreasing in the first argument, i.e., for $y_1 \le y_2$ we
have $|t^+(y_1,\vec x)| \le |t^+(y_2,\vec x)|$, such that for all $y,\vec
x$, $t(y,\vec x) \le t^+(y,\vec x)$. For instance:
\begin{align*}
t^+(\epsilon ,\vec x) & := t(\epsilon, \vec x), \\
t^+(y',\vec x) & := 
\begin{cases}
t(y',\vec x) & \text{if $|t(y,\vec x)| \le |t(y',\vec x)|$}, \\
t(y,\vec x) & \text{otherwise}.
\end{cases}
\end{align*}
In fact, if $t$ is itself non-decreasing in the first argument, then
$t^+$ is equal to $t$. 

Now, we get that 
$$\MBPR(g,h,t) = \PR(g,(h|_t)^m) = \BPR(g,(h|_t)^m,t^+).$$
\end{remark}

\begin{remark}\label{rht}
\begin{enumerate}
\item If $h,t \in [\II; \CC, \BRN, \MBPR]$ (or $[\BB;\PC, \PRN, \MPPR]$),
  then we have also $h|_t \in [\II; \CC, \BRN, \MBPR]$ (or $[\BB;\PC, \PRN, \MPPR]$, respectively). 
\item If $h \in [\II; \CC, \BRN, \MBPR]$ (or $[\BB;\PC, \PRN, \MPPR]$), then
we have $h^m \in [\II; \CC, \BRN, \MBPR]$ (or $[\BB;\PC, \PRN, \MPPR]$, respectively).
\end{enumerate}
Moreover, the function definitions of $h|_t$ and $h^m$ do not make any
extra use of the $\MBPR$ (or $\MPPR$ respectively) scheme (relatively to the definitions of $h$ and
$t$). 
\end{remark}

Define by bounded recursion on notation $P(\epsilon) = \epsilon$ and
$P(xi) = x|_x$ and $D(\epsilon,x) = x$ and $D(yi,x) = P(x)|_x$. Then
$x|_y = D(D(y,x),x)$. This justifies item (i) of the remark
above. Item (2) is an obvious consequence of $\preceq$ being decidable
in $\textsc{P}$. The case of $[\BB;\PC, \PRN, \MPPR]$ is similar.

%\begin{definition}
%$h$ is $t$-bounded and monotone if $(h|_t)^m = h$.
%\end{definition}

\begin{theorem}
$[\II;\CC, \BRN, \MBPR] = \FPH.$
\end{theorem}
\begin{proof}
We prove that 
\begin{enumerate}
\item for all $f \in [\II;\CC, \BRN, \MBPR]$ there exists a $ F \in
  [\BB;\PC, \PRN, \MPPR]$ such that $\forall \vec x. f(\vec x) =
  F(\vec x;)$;
\item for all $F \in [\BB;\PC, \PRN, \MPPR]$ there exists a $f \in
  [\II;\CC, \BRN, \MBPR]$ such that $\forall \vec x, \vec y. F(\vec x;
  \vec y) = f(\vec x, \vec y)$.
\end{enumerate}
This shows that $[\II;\CC, \BRN, \MBPR]$ and $[\BB;\PC, \PRN, \MPPR]$
can be identified. Thus, the present statement is a consequence of
Proposition~\ref{p6}.

(1) is proven by induction on the complexity of the function
definitions. The proof is analogous to the proof of Theorem 3.2 in
\cite[p.~121]{Oit97}. It uses remark \ref{rht}. 
%, it suffices to
%consder applications of $MBPR(g,h,t)$ where $h$ is $t$-bounded and
%monotone. 

The proof of (2) is straightforward, by induction on the complexity of
the function definition of $F \in  [\BB;\PC, \PRN, \MPPR]$. It uses
remark \ref{r7}(2). Obviously, the $\BB$ functions (4)--(6) are
defined using bounded recursion on notation.
\end{proof}

\section{The theory $\aph$}

The applicative theory $\aph$ is based on the basic theory $\Bb$ of
operations and words, as introduced by Strahm in \cite[\S~3.1]{Str03}, with
slight modifications indicated below. In particular, our application
is total, while Strahm works in a partial setting.

We formulate $\Bb$ in a standard first order language, with
\emph{individual variables} $x,y,z,\dots$, \emph{individual
  constants}: $\ka, \ess$ (combinators); $\textsf{p}, \textsf{p}_0,
\textsf{p}_1$ (pairing and projection); $\cw$ (case distinction);
$\epsilon$ (empty word); $\textsf{s}_0, \textsf{s}_1$ (binary
successors), $\textsf{p}_\We$ (binary predecessor); $\textsf{s}_\ell,
\textsf{p}_\ell$ (lexicographic successor and predecessor); $\csubs$
(initial subword relation); $*, \times$ (word concatenation and word
multiplication). There is one binary function symbol $\cdot$ for term
application, which, however, is usually written by juxtaposition. We
have only one unary relation symbol $\We$ (binary words), and one
binary relation symbol $=$ (equality). \emph{Terms} ($r,s,t,\dots$)
are build from variables and constants by term application. 

We use the usual abbreviations of the framework of applicative
theories, which include, in particular, the following
ones:\rrandnotiz{$\not=$; $f(x,y)$ vs. $f\,x\,y$}
\begin{align*}
0 & := \sz{\epsilon}, \\
1 & := \so{\epsilon}, \\
s \subseteq t & := \csubs\,s\,t = 0, \\
s \le t & := \lw\,s \subseteq \lw\,t, \\
s * t & := *\,s\,t, \\
s \times t & := \times\,s\,t. 
\end{align*}
As we will define $\lw\,t$ by $1 \times t$, $s \le t$ stands actually
for $1 \times s \subseteq 1 \times t$.\footnote{Note that, in $\aph$
  the relation $\le$ compares the lengths of the terms, while we used
  the same symbol before, outside $\aph$, to compare the terms themselves.} For $w \in \WW$, $\overline{w}$
is the corresponding applicative term. 

\emph{Formulas} are usual first-order formulas, build from the atomic
formulas $\We(t)$ and $t=s$ by use of negation ($\neg$), conjunction
($\wedge$), disjunction ($\vee$), implication ($\to$), and universal
($\forall x$) and existential ($\exists x$) quantification. 
As abbreviation we use
\begin{align*}
\forall x \in \We. \phi & := \forall x. \We(x) \to \phi,\\
\exists x \in \We. \phi & := \exists x. \We(x) \wedge \phi,\\
\exists x \le t. \phi & := \exists x \in \We. x \le t \wedge \phi,\\
t : \We \to \We & := \forall x \in \We. \We(t\,x),\\
t : \We^2 \to \We & := \forall x \in \We. \forall y \in \We. \We(t\,x\,y).
\end{align*}

Note that Strahm formulates $\Bb$ within the \emph{logic of partial
  terms}, which includes an extra existence
predicate. However, for
the present purpose, partiality is not essential and hence we stick to
total application. Thus, our logic is
standard, \emph{classical} first order logic. For more background on
applicative theories see, for instance, \cite{Bee85}, \cite{JKS99}, or \cite{Kah07}.

The non-logical axioms of $\Bb$ are the following ones:\footnote{In
  \cite{Str03}, Strahm axiomatizes also the tally length of binary
  words, $\lw$, since his theory $\Bb$ does not include word
  concatenation and word multiplication from the very beginning. In
  the presence of word multiplication the tally length can be defined
  by letting $\lw\,t = 1 \times t$.}

\begin{enumerate}
\item[I.] Combinatory algebra and pairing
\begin{enumerate}
\item[(1)] $\ka\,x\,y = x$,
\item[(2)] $\ess\,x\,y\,z = x\,z\,(y\,z)$,
%\randnotiz{partial? Thomas is partial!}
\item[(3)] $\textsf{p}_0(\textsf{p}\,x\,y) = x \wedge
\textsf{p}_1(\textsf{p}\,x\,y) = y$.
\end{enumerate}
\item[II.] Definition by cases on $\We$.\footnote{Our case distinction
    checks the last bit of a word, while Strahm uses a case
    distinction which compares words as a whole.}\rrandnotiz{Comparing
    in detail?}
\begin{enumerate}
\item[(4)] $\cw\,\epsilon\,s\,r\,u = s,$
\item[(5)] $\We(t) \to \cw\,(\sz{t})\,s\,r\,u = r,$
\item[(6)] $\We(t) \to \cw\,(\so{t})\,s\,r\,u = u,$
\end{enumerate}
\item[III.] Closure, binary successors, and predecessors
\begin{enumerate}
\item[(7)] $\We(\epsilon) \wedge \forall x. \We(x) \to \We(\sz{x})
  \wedge \We(\so{x}),$
\item[(8)] $\sz{x} \not= \so{x} \wedge \sz{x} \not= \epsilon \wedge \so{x}
  \not= \epsilon,$
\item[(9)] $\textsf{p}_\We: \We \to \We \wedge \pw{\epsilon} = \epsilon,$
\item[(10)] $\We(x) \to \pw{(\sz{x})} = x \wedge \pw{(\so{x})} = x,$
\item[(11)] $\We(x) \wedge x \not= \epsilon \to \sz{(\pw{x})} = x \vee
  \so{(\pw{x})} = x.$
\end{enumerate}
\item[IV.] Lexicographic successor and predecessor
\begin{enumerate}
\item[(12)] $\textsf{s}_\ell : \We \to \We \wedge \sll{\epsilon} = 0,$
\item[(13)] $\We(x) \to \sll{(\sz{x})} = \so{x} \wedge \sll{(\so{x})}
      = \sz{(\sll{x})},$
\item[(14)] $\textsf{p}_\ell : \We \to \We \wedge \sll{\epsilon} = \epsilon,$
\item[(15)] $\We(x) \to  \pll{(\sll{x})} = x,$
\item[(16)] $\We(x) \wedge x \not= \epsilon \to \sll{(\pll{x})} = x.$
\end{enumerate}
\item[V.] Initial subword relation
\begin{enumerate}
\item[(17)] $\We(x) \wedge \We(y) \to \csubs\,x\,y = 0 \vee \csubs\,x\,y = 1$,
\item[(18)] $\We(x) \to (x \subseteq \epsilon \leftrightarrow x = \epsilon)$,
\item[(19)] $\We(x) \wedge \We(y) \wedge y \not= \epsilon \to 
(x \subseteq y  \leftrightarrow x \subseteq \pw{y} \vee x = y),$
\item[(20)] $\We(x) \wedge \We(y) \wedge \We(z) \wedge x \subseteq y
  \wedge y \subseteq z \to x \subseteq z.$
\end{enumerate}
%\item[VI.] Tally length of binary words
%\begin{enumerate}
%\item[(20)]
%\item[(21)]
%\item[(22)]
%\item[(23)]
%\item[(24)]
%\end{enumerate}
\item[VI.] Word concatenation
\begin{enumerate}
\item[(21)] $*: \We^2 \to \We$,
\item[(22)] $\We(x) \to x*\epsilon = x$,
\item[(23)] $\We(x) \wedge \We(y) \to 
x * (\sz{y}) = \sz{(x*y)} \wedge 
x * (\so{y}) = \so{(x*y)}$.
\end{enumerate}
\item[VII.] Word multiplication
\begin{enumerate}
\item[(24)] $\times: \We^2 \to \We$,
\item[(25)] $\We(x) \to x \times \epsilon = \epsilon$,
\item[(26)] $\We(x) \wedge \We(y) \to x \times \sz{y} = (x \times y)*
  x \wedge x \times \so{y} = (x \times y)* x$.
\end{enumerate}
\end{enumerate}

\paragraph{Induction on notation.}

$$f: \We \to \We \wedge \phi(\epsilon) \wedge (\forall x \in
\We. \phi(x) \to \phi(\sz{x}) \wedge \phi(\so{x})) \to \forall x \in
\We. \phi(x),$$
where $\phi(x)$ is of the form $\exists y \le f\,x.\psi(f,x,y)$ for
$\psi(f,x,y)$ a \emph{positive and $\We$-free}
formula.\footnote{Positive formulas are defined, as usual, as negation
  and implication free formulas.}

This induction is called $\sbwiw$ in \cite{Str03}.

\paragraph{Monotonicity relation.}

%We now define the monotonicity relation within the applicative
%framework. It is introduced via a characteristic function $\cmono$,
%defined by recursion.

%\begin{definition}[Monotonicity relation]
%%Within the applicative framework, we define a characteristic function
%%corresponding to the monotonicity defined above, writing $s \succeq
%%t$ for $\succeq\,s\,t$:
%$$\cmono\,s\,t = \begin{cases}
%0 & \text{if $s < t$}\\
%1 & \text{if $t < s$}\\
%0 & \text{if $t = \epsilon \wedge s = \epsilon$}\\
%\cmono\,(\pll{s})\,(\pll{t}) \wedge \cw(s,0,\cw(t,\epsilon,0,1),0)& \text{otherwise}
%\end{cases}$$
%\end{definition}\rrandnotiz{Write down as applicative term?}

It is easy to observe that the monotonicity relation $\preceq$ is
polytime decidable. As the theory $\Bb+\sbwiw$ allow to represent all
polytime functions (as provably total functions in the sense of
Definition \ref{provtotal} below), we know that there is term
$t_{\chi_\preceq}$ with 
\begin{enumerate}
\item $\Bb+\sbwiw \vdash
  t_{\chi_\preceq}\,\overline{w_1}\,\overline{w_2} =
  \overline{\chi_\preceq(w_1,w_2)}$, for all $w_1,w_2 \in \WW$, and
\item $\Bb+\sbwiw \vdash \forall x,y. \We(x) \wedge \We(y) \to t_{\chi_\preceq}\,x\,y = 0 \vee t_{\chi_\preceq}\,x\,y = 1.$
\end{enumerate}

In the following, we will use $\cmono$ as abbreviation for $\lambda
x,y. t_{\chi_\preceq}\,y\,x$. Moreover, $s \succeq t$ is used as
abbreviation of $\cmono\,s\,t = 0$.  
We
also introduce quantifier $\exists x \succeq t. \phi$ as abbreviation
for $\exists x. \We(x) \wedge x \succeq t \wedge \phi$.

%By use of induction on notation, one can show that $\cmono$ is total
%on words, in the sense that:\rrandnotiz{Add proof?}
%$$\forall x,y. \We(x) \wedge \We(y) \to \cmono\,x\,y = 0 \vee
%\cmono\,x\,y = 1.$$
Note that 2.\ above means that $\cmono$ is total as function from $\We
^2 \to \We$. But, of course, $\cmono$ is not total as a binary
relation, as we have, for 
instance, $01 \not\succeq 10$ and $10 \not\succeq 01$. 

\begin{remark}\label{lepreceq} For $u$ and $v$ in $\We$, we can show in $\aph$:
\begin{enumerate}
\item $u \le v \to u \preceq 1 \times v$,
\item $u \preceq v \to u \le v$.
\end{enumerate}
And we can define a low-level pairing function $\langle
\cdot,\cdot\rangle$ and projections $(\cdot)_0$ and $(\cdot)_1$ on
$\We$, which are, at most, in $\textsc{FPtime}$, such that $\aph$ proves
for the
representing terms:\randnotiz{Referenz? II.6}
\begin{enumerate}
\item[3.] $u \preceq \langle u,v\rangle$ and $v \preceq \langle
  u,v\rangle$,
\item[4.] $(u)_0 \preceq u$ and $(u)_1 \preceq u$.
\end{enumerate}
\end{remark}

\paragraph{Monotone induction $\sbwmpi$.}
\begin{multline*}
t: \We \to \We \wedge 
(\exists x\in \We.\phi(\epsilon,x)) \wedge %\\
(\forall y \in \We.\forall x\in \We. \phi(y,x) \to 
\exists z \succeq x. \phi(\sll{y},z))
\to \\
\forall y \in \We. \exists x\in \We. \phi(y,x),
\end{multline*}
% where $\phi(x,n)$ is of the form $\exists k \le t\,x\,n.\psi(f,x,n)$ for
%$\psi(f,x,n)$ a \emph{positive and $\We$-free} formula.
where $\phi(y,x)$ is of the form $x \le t\,y \wedge \psi(t,y,x)$ for
$\psi(t,y,x)$ a \emph{positive and $\We$-free} formula \emph{not
  containing disjunctions}. For the reason of the exclusion of
disjunctions, see remark \ref{dis} below.\rrandnotiz{Remark about
  heuristics; is this scheme intuitionistically valid??}

Essentially, $\aph$ is equal to Strahm's theory
  $\textsf{PT}$ plus the monotone induction scheme $\sbwmpi$.

\section{The lower bound}

\begin{definition}\label{provtotal}
A function $F: \WW^n \to \WW$ is called \emph{provably total in
  $\aph$}, if there exists a closed term $t_F$ such that
\begin{enumerate}
\item $\aph \vdash t_F\,\overline{w_1}\,\dots\,\overline{w_n} =
  \overline{F(w_1,\dots,w_n)}$ for all $w_1,\dots,w_n \in \WW$, and
\item $\aph \vdash t_F: \We^n \to \We$.
\end{enumerate}
\end{definition}

%\begin{remark} The characteristic function $\chi_\preceq$ of the
%  monotonicity relation defined over binary words is provably total in
%  $\aph$ with $t_{\chi_\preceq} := \cmono$. 
%\end{remark}

Using the result of \cite[\S~4]{Str03} about the provably total
function in Strahm's theory corresponding to \textsc{FPtime}, it
remains to show that functions defined by the monotone bounded
primitive recursion scheme $\MBPR(g,h,t)$ are provably total in
$\aph$. 

So, let us assume that $g$, $h$, and $t$ are provably total in
$\aph$, and $f$ be defined as $\MBPR(g,h,t) = \PR(g,(h|_t)^m)$.
%, i.e., 
%\begin{align*}
%f(\epsilon, \vec x) & = g(\vec x),\\
%f(y',\vec x) & = h^m(y,\vec x,f(y,\vec x))|_{t(y,\vec x)}
%= \begin{cases}
%h(y, \vec x, f(y, \vec x)) & \text{if $f(y,\vec x) \preceq h(y, \vec
%  x, f(y, \vec x))$} \\
%f(y,\vec x) & \text{otherwise}.
%\end{cases}
%\end{align*}

Now, in $\aph$, let 
\begin{align*}
f(\epsilon,\vec z) & = g(\vec z) \\
f(\sll{y},\vec z) & =
\begin{cases}
%%%t(y, \vec x) & \text{if $h(y, \vec x,f(y,\vec x)) > t(y, \vec x)$},\\
h|_t(y, \vec z,f(y,\vec z)) & \text{if $f(y,\vec z) \preceq h|_t(y,\vec
  z,f(y,\vec z))$} \\
f(y,\vec z) & \text{otherwise}
\end{cases}
\end{align*}
and we show by monotone induction that $\forall y \in \We. \exists
x\in \We. x \le
t_f(y,\vec z) \wedge f(y,\vec z) = x$%.
% (in fact, to match the formal
%requirement of the structure of the induction formula, we would have
%to use the equivalent formula: $\forall y. \exists n. \exists m \le
%t_f(y,\vec z) \wedge n = m \wedge f(y,\vec z) = n$)
, where $$t_f(y,\vec
x) = \begin{cases} g(\vec z) & \text{if $y = \epsilon$,} \\
t^+(y,\vec z) & \text{otherwise.}
\end{cases}$$

\emph{Induction base:}
As $f(\epsilon,\vec z) = g(\vec z)$, and $g$ is provably total in
$\aph$, we have $\exists x\in \We. x \le g(\vec z) \wedge f(\epsilon,\vec z) = x$.

\emph{Induction step:}
%We have to show that $\forall y.\forall n. (\exists m_0 \le t_f(y,\vec z)
%\wedge n = m_0 \wedge f(y,\vec z) = n) \to
%\exists n_1 \succeq n. \exists m_1 \le t_f(\sll{y},\vec z) \wedge m_1 = n_1
%\wedge f(\sll{y},\vec z) = n_1$. 
\ We have to show that $\forall y\in \We.\forall x\in \We. x \le t_f(y,\vec z)
\wedge f(y,\vec z) = x \to
\exists x_1 \succeq x.\linebreak[4] x_1 \le t_f(\sll{y},\vec z) \wedge
f(\sll{y},\vec z) = x_1$.  

By definition, 
\begin{align*}
f(\sll{y},\vec z) & =
\begin{cases}
%%%t(y, \vec z) & \text{if $h(y, \vec z,f(y,\vec z)) > t(y, \vec z)$},\\
h|_t(y, \vec z,f(y,\vec z)) & \text{if $f(y,\vec z) \preceq h|_t(y,\vec
  z,f(y,\vec z))$}, \\
f(y,\vec z) & \text{otherwise}.
\end{cases}
\end{align*}

In the first case, the assertion follows immediately from the
condition $f(y,\vec z) \preceq h|_t(y,\vec z,f(y,\vec z))$.

In the second case, the assertion follows immediately from the premise
(choosing $x_1 := x$). 

Thus, we can conclude by monotone induction that 
$\forall y \in \We. \exists x\in \We. x \le
t_f(y,\vec z) \wedge f(\sll{y},\vec z) = x$. 

Thus, we get the following result:
\begin{lemma}
The provably total functions of $\aph$ include $\FPH$.
\end{lemma}

\section{The upper bound}

The proof of the upper bound follows quite closely the proof of the
upper bound of Strahm for his theory $\textsf{PT}$ in
\cite[\S~6]{Str03}. For it, one reformulates the theory first in
Gentzen's classical sequence calculus, and proves partial cut
elimination, such that the remaining cuts are restricted to positive
formulas. In a second step, one realizes positive derivations with
realizers from the appropriate complexity class. In this step, one
uses the open term model $\mathcal{M}(\lambda\eta)$ of the applicative
ground structure, which is based on the usual $\lambda\eta$ reduction
of the untyped $\lambda$-calculus. In fact, $\eta$ allows us to treat
extensionality of operations, i.e., we may add the following axiom to
$\aph$:
\begin{itemize}
\item[(\textsf{Ext})\hspace*{-10pt}] \hspace*{10pt}$\forall f, g. (\forall x. f\,x = g\,x) \to f = g$.
\end{itemize}

For the treatment of $\aph$, we will follow Strahm's proof for
$\textsf{PT}$, and check only, how to take care of our additional
monotone induction scheme $\sbwmpi$.

Let $\aphp$ the Gentzen-style sequent calculus reformulation of
$\aph$ such that all main formulas of non-logical axioms and rules are
positive. 
In this calculus, the monotone induction $\sbwmpi$ is
rewritten as the following rule:

\newcommand{\seqq}[2]{\seq{\Gamma,#1}{#2,\Delta}}

$$
\frac{\stackrel{\displaystyle 
\stackrel{\displaystyle \seqq{\We(u)}{\We(t\,u)}}{
\seq{\Gamma}{\exists n. \We(n) \wedge \phi(\epsilon,n),\Delta}}}{
\seqq{\We(a),\We(b),\phi(a,b)}{\exists m \succeq b.\phi(\sll{a},m)}
}}%
{\seqq{\We(s)}{\exists n. \We(n) \wedge
    \phi(s,n)}},$$
where $\phi(s,n)$ is of the form $n \le t\,s \wedge \psi(t,s,n)$ for
$\psi(t,s,n)$ a positive and $\We$-free formula which does not contain
disjunctions. 

We write $\aphp \vdash \seq{\Gamma}{\Delta}$ if the sequent
$\seq{\Gamma}{\Delta}$ is derivable in $\aphp$, and $\aphp
\vdash_{\!\!\!*} \seq{\Gamma}{\Delta}$ if it has a proof where all cut
formulas are \emph{positive}.

\subsection{Partial cut elimination}

\begin{theorem}[Partial cut elimination, cf.~{\cite[Theorem~12]{Str03}}]
  For all sequents $\seq{\Gamma}{\Delta}$, $\aphp \vdash
  \seq{\Gamma}{\Delta}$ implies $\aphp \vdash_{\!\!\!*}
  \seq{\Gamma}{\Delta}$.
\end{theorem}

We only have to check that the main formulas of our induction rules
are positive, but that is the case since, in particular, $\exists m
\succeq b.\phi(\sll{a},m)$ is positive.\rrandnotiz{Maybe more explicit...}

\begin{corollary}[cf.~{\cite[Corollary~13]{Str03}}] If
  $\seq{\Gamma}{\Delta}$ is a sequent of positive formulas with $\aphp
  \vdash \seq{\Gamma}{\Delta}$, then there is a $\aphp$ derivation of
  $\seq{\Gamma}{\Delta}$ which contains only positive formulas.
\end{corollary}

\subsection{Realizability}

\begin{definition}%[cf.~{\cite[page 28]{Str03}}]
Let $\rho \in \WW$ and $\phi$ a positive formula. Then
$\rr{\rho}{\phi}$ is inductively defined as follows:\footnote{Here
  $\langle \cdot,\cdot \rangle$ is a low-level pairing function on
  binary words, with its projections $(\cdot)_0$ and $(\cdot)_1$.}
\begin{eqnarray*}
\rr{\rho}{\We(t)} & \quad \textrm{if} \quad  & \MMle \models t = \overline{\rho}, \\
\rr{\rho}{{(t_1 = t_2)}} & \quad \textrm{if} \quad & \rho = \epsilon
 \textrm{ and } \MMle \models t_1 = t_2, \\
\rr{\rho}{{(\phi \wedge \psi)}} & \quad \textrm{if} \quad &
 \rho = \langle \rho_0,\rho_1\rangle \text{ and }
  \rr{\rho_0}{\phi} \text{ and } \rr{\rho_1}{{\psi}}, \\
\rr{\rho}{{(\phi \vee \psi)}} & \quad \textrm{if} \quad  & 
 \rho = \langle i,\rho_0\rangle \text{ and either $i = 0$ and
   $\rr{\rho_0}{\phi}$ or $i = 1$ and $\rr{\rho_0}{\psi}$}, \\
\rr{\rho}{{(\forall x. \phi(x))}} & \quad \textrm{if} \quad  & 
 \rr{\rho}{\phi(u)} \text{ for a fresh variable $u$},\\
\rr{\rho}{{(\exists x. \phi(x))}} & \quad \textrm{if} \quad  &
 \rr{\rho}{\phi(t)} \text{ for some term $t$}.
\end{eqnarray*}
\end{definition}

$\rho$ realizes a sequence $\Delta$ of $n$ formulas
$\phi_1,\dots,\phi_n$, if $\rho = \langle i_2,\rho_0\rangle$, $1 \le i
\le n$, $i_2$ the dyadic representation of the natural number $i$, and $\rr{\rho_0}{\phi_i}$.

To improve readability, we use the following abbreviations regarding
our low-level pairing in the context of realizability:
When we $\rho$ realizes a conjunction $\phi \wedge \psi$,
$\pleft{\rho}$ for the $(\rho)_0$, i.e., the realizer of $\phi$, and,
analogously $\pright{\rho}$ for the realizer $(\rho)_1$ of
$\psi$. When $\rho$ realizes a sequence $\phi_1,\dots,\phi_n$, we
write $\pnumber{\rho}$ for $(\rho)_0$, i.e., the index of the realized
formula, and $\pselect{\rho}$ for $(\rho)_1$, the realizer of the
selected formula.

\begin{theorem}[Realizability for $\aphp$, cf.~{\cite[Theorem~15]{Str03}}]
Let $\seq{\Gamma}{\Delta}$ be a sequent of positive formulas with
$\Gamma = \phi_1,\dots,\phi_n$ and assume that $\aphp \vdash_{\!\!\!*}
\seq{\Gamma[\vec u]}{\Delta[\vec u]}$. Then there exists a function
$F: \WW^n \to \WW$ in $\FPH$ such that for all terms $\vec{s}$
and all $\rho_1,\dots,\rho_n \in \WW$:
$$\rr{\rho_1}{\phi_1[\vec{s}]},\dots,\rr{\rho_n}{\phi_n[\vec{s}]} \qquad
\Longrightarrow \qquad \rr{F(\rho_1,\dots,\rho_n)}{\Delta[\vec s]}.$$
\end{theorem}

The proof runs by induction on the length of a quasi cut-free
derivation. We have only to check the case of our monotone induction
rule, as all other cases are like in \cite{Str03}.

By induction hypothesis, we get for the three premises:
\begin{align}
\Gamma, \We(u) & \Rightarrow \We(t\,u),\Delta \\
\Gamma & \Rightarrow \exists n. \We(n) \wedge \phi(\epsilon,n),\Delta
\\
\Gamma, \We(a),\We(b),\phi(a,b) &\Rightarrow \exists m \succeq b.\phi(\sll{a},m),\Delta
\end{align}
%
%\begin{align}
%\Gamma, \We(a),\We(b),\phi(a,b) &\Rightarrow \exists m. \We(m) \wedge
%m \succeq b \wedge \phi(\sll{a},m),\Delta
%\end{align}
%
that there are functions $T$, $G$ and $H$ in $\FPH$ such
that\rrandnotiz{$\vec s$?!}
 for all $\vec\rho, \sigma,\tau,\upsilon$:\rrandnotiz{Footnote 3 in Strahm p.32}
\begin{align}
\rr{\vec{\rho}}{\Gamma[\vec{s}]} &\quad \Rightarrow\quad
  \rr{T(\sigma,\vec\rho)}{\We(t[\vec s](\sigma)), \Delta[\vec s]} \nonumber\\
\rr{\vec{\rho}}{\Gamma[\vec{s}]} &\quad \Rightarrow\quad
  \rr{G(\vec\rho)}{\exists n. \We(n) \wedge
    \phi(\epsilon,n)[\vec{s}],\Delta[\vec{s}]} \label{Gr}\\
%\rr{\vec{\rho}}{\Gamma[\vec{s}]},
%\rr{\sigma}{\We(a)},
%\rr{\tau}{\We(b)},
%\rr{\upsilon}{\phi(a,b)[\vec{s}]}
% &\quad \Rightarrow\quad
%  \rr{H(\vec\rho,\sigma,\tau,\upsilon)}{\exists m \succeq b.\phi(\sll{a},m)[\vec{s}],\Delta[\vec{s}]}
%%\end{align*}
%\\
%\begin{align*}
\rr{\vec{\rho}}{\Gamma[\vec{s}]},
\rr{\upsilon}{\phi(\sigma,\tau)[\vec{s}]}
 &\quad \Rightarrow\quad
  \rr{\tilde{H}(\sigma,\vec\rho,\tau,\upsilon)}{\exists m \succeq
    \tau.\phi(\sll{\sigma},m)[\vec{s}],\Delta[\vec{s}]} \label{Hr}
%\\
%\rr{\vec{\rho}}{\Gamma[\vec{s}]},
%\rr{\upsilon}{\phi(\sigma,\tau)[\vec{s}]}
% &\quad \Rightarrow\quad
%  \rr{H(\sigma,\vec\rho,\tau,\upsilon)}{\exists m. \We(m) \wedge m
%    \succeq \tau \wedge \phi(\sll{\sigma},m)[\vec{s}],\Delta[\vec{s}]}
\end{align}

Now, we need a function $F$ in $\FPH$, such that 
\begin{align}
%\rr{\vec\rho}{\Gamma[\vec s]}, \rr{\sigma}{\We(t)} \quad\Rightarrow\quad
%  \rr{F(\vec\rho,\sigma)}{\exists n. \We(n) \wedge \phi(t,n)[\vec
%  s],\Delta[\vec s]}
%%\end{align}
%\\
%\begin{align}
\label{F}
\rr{\vec\rho}{\Gamma[\vec s]} \quad\Rightarrow\quad
  \rr{F(\sigma,\vec\rho)}{\exists n. \We(n) \wedge \phi(\sigma,n)[\vec
  s],\Delta[\vec s]}
\end{align}

We set
\begin{multline*}
H(\sigma,\vec\rho,\omega) =
\langle 1, 
\langle \pleft{\pselect{\tilde{H}(\sigma,\vec\rho,\pleft{\omega},\pright{\omega})}},\\
\pright{\pright{\pselect{\tilde{H}(\sigma,\vec\rho,\pleft{\omega},\pright{\omega})}}}\rangle\rangle.
\end{multline*}
% \tilde{H}(\sigma,\vec\rho,(\omega)_0,(\omega)_1)$. 

This definition looks quite involved, its idea is, however,
straightforward: when, according to (\ref{Hr}), $\tilde{H}$ will realize
a formula of the form $\exists m \succeq
\tau.\phi(\sll{\sigma},m)[\vec{s}]$, $H$ is supposed to realize
$\exists m. \We(m) \wedge \phi(\sll{\sigma},m)[\vec{s}]$. Thus we have
to ``cut out'' the second conjunct $m \succeq \tau$ under the
existential quantifier ($\We(m)$ is the first conjunct which is not
visible in the abbreviation $\exists m \succeq \tau$).

Before defining the function $F$ which should realize the conclusion
of our rule, we define an auxiliary function $F'$ which returns a
pair, having the intended value of $F$ as its second component. The
first component serves only to guarantee the monotonicity. 

So, $F'(\sigma,\vec\rho,\tau)$ is defined by monotone recursion as:
%$MBPR(G,H,T) = BPR(G,H^m,T)$, i.e., 
\begin{align*}
F'(\epsilon,\vec\rho) & = \langle\epsilon,G(\vec\rho)\rangle,\\
F'(\sll{\sigma},\vec\rho) & = 
\begin{cases}
F'(\sigma,\vec\rho) & \text{if $\pnumber{\pright{F'(\sigma,\vec\rho)}} \not=
  1$} \\ & \qquad\text{($F$ will realize one of the $\Delta$s),}\\
%1 \times T(\sigma,\vec\rho) & \text{if $\pnumber{F(\sigma,\vec\rho)} =
%  1$ and $\pnumber{T(\sigma,\vec\rho)} \not= 1$}. 
\langle %\pleft{
F'(\sigma,\vec\rho)%}
,T(\sigma,\vec\rho)\rangle
& \text{if $\pnumber{\pright{F'(\sigma,\vec\rho)}} =
  1$ and $\pnumber{T(\sigma,\vec\rho)} \not= 1$}
\\ & \qquad \text{($T$ realizes one of the $\Delta$s),}\\
\langle \epsilon,H(\sigma,\vec\rho,\pselect{\pright{F'(\sigma,\vec\rho)}})\rangle & \text{otherwise.}
%%\langle 1, \langle (H(\sigma,\vec\rho,(F(\sigma,\vec\rho))_1))_0,
%%((H(\sigma,\vec\rho,(F(\sigma,\vec\rho))_1))_1)_1\rangle\rangle
 %|_{T(\sigma,\vec\rho)}
\end{cases}
\end{align*}

With this function, $F(\sigma,\vec\rho)$ is defined as $\pright{F'(\sigma,\vec\rho)}$.

To check (\ref{F}) we can use a straightforward (meta-)induction on
$\sigma$:

$\sigma = \epsilon$: Given $\rr{\vec\rho}{\Gamma[\vec s]}$, in this
case, $F(\epsilon, \vec\rho) = \rr{G(\vec\rho)}{\exists n. \We(n)
\wedge \phi(\epsilon,n)[\vec s], \Delta[\vec s]}$ by (\ref{Gr}).

$\sll{\sigma}$: In the first and second case, we know that one of the side
formulas $\Delta[\vec s]$ is realized, and, of course, $F(\sll{\sigma},
\vec\rho)$ realizes one of these side formulas, too.
In the third case, we have to show that 
$$\rr{H(\sigma,\vec\rho,\pselect{F(\sigma,\vec\rho)})}{\exists
  n. \We(n) \wedge \phi(\sll{\sigma},n)[\vec s], \Delta[\vec s]}.$$  

We know that $\pnumber{F(\sigma,\vec \rho)} = 1$, thus, using the
induction 
hypothesis, we know that the first formula of the sequence is
realized, i.e., 
$$\rr{\pselect{F(\sigma,\vec\rho)}}{\exists n. \We(n)
\wedge \phi(\sigma,n)[\vec s]}.$$ 
That means, $\pleft{\pselect{F(\sigma),\vec \rho}} = \tau$ for a $\tau$ with $\rr{\pright{\pselect{F(\sigma),\vec
\rho}}}{\phi(\sigma,\tau)[\vec s]}$. 
By definition of
$H(\sigma,\vec\rho,\pselect{F(\sigma,\vec\rho)})$ is
$\tilde{H}(\sigma,\vec\rho,\pleft{\pselect{F(\sigma,\vec\rho)}},\pright{\pselect{F(\sigma,\vec\rho)}})$.
Letting $\tau$ be as above the term $\pleft{\pselect{F(\sigma),\vec
\rho}}$, and $\upsilon := \pright{\pselect{F(\sigma,\vec\rho)}}$, we get from
(\ref{Hr}) that
\begin{multline*}H(\sigma,\vec\rho,\pselect{F(\sigma,\vec\rho)})m = \\
\rr{\tilde{H}(\sigma,\vec\rho,\pleft{\pselect{F(\sigma,\vec\rho)}},\pright{\pselect{F(\sigma,\vec\rho)}})}{\exists
  m \succeq
  \pleft{\pselect{F(\sigma,\vec\rho)}}.\phi(\sll{\sigma},m)[\vec{s}],\Delta[\vec{s}]}.
\end{multline*}
The remaining coding serves to get rid of the redundant monotonicity condition.

It remains to show that $F$ is in $\FPH$. For it, we only need to check
that the step function $F'$ is of the form $h|_t$, with $h$ and $t$ in
$[\II;\CC, \BRN, \MBPR]$, and monotone. 

That the step function is bounded follows essentially as in the proof
of \cite[Theorem~15]{Str03} with the fact that the formula $\phi(y,n)$
has the shape $n \le t\,y \wedge \psi(t,y,n)$.

Monotonicity: as in the first and second case, the function stays
constant, we only have to check that the value is greater or equal (in
the sense of our monotonicity relation $\preceq$) as the recursive
argument $F'(\sigma,\vec \rho)$. This is trivial in the first case
(where it is equal), and follows in the second case from the fact that
$F'(\sigma,\vec \rho)$ is coded in the first argument of the pair.
In the third case, we have to show that, for all $\sigma$,
$F'(\sigma,\vec\rho) \preceq
\langle \epsilon,
H(\sigma,\vec\rho,\pselect{\pright{F'(\sigma,\vec\rho)}})\rangle.$
From the case 
distinction, we know, that $\pright{F'(\sigma,\vec\rho)} = \langle 1,
\pselect{\pright{F'(\sigma,\vec\rho)}}\rangle$, and
$\rr{\pselect{\pright{F'(\sigma,\vec\rho)}}}{\exists n. \We(n) \wedge
  \phi(\sigma,n)[\vec s]}$, i.e., $\pselect{\pright{F'(\sigma,\vec\rho)}}$ is of
the form $\langle \omega_0, \omega_1 \rangle$ with
$\rr{\omega_1}{\phi(\sigma,\omega_0)[\vec s]}$. On the other hand,
\begin{align*}
& \hphantom{{} = {}\ }H(\sigma,\vec\rho,\pselect{\pright{F'(\sigma,\vec\rho)}})\\
& = 
\langle 1, 
\langle
\pleft{\pselect{\tilde{H}(\sigma,\vec\rho,\pleft{\pselect{\pright{F'(\sigma,\vec\rho)}}},\pright{\pright{\pselect{\pright{F'(\sigma,\vec\rho)}}}})}},\\
&\qquad\qquad
\pright{\pright{\pselect{\tilde{H}(\sigma,\vec\rho,\pleft{\pselect{\pright{F'(\sigma,\vec\rho)}}},\pright{\pselect{\pright{F'(\sigma,\vec\rho)}}})}}}\rangle\rangle
\\
& =
\langle 1, 
\langle \pleft{\pselect{\tilde{H}(\sigma,\vec\rho,\omega_0,\omega_1)}},
\pright{\pright{\pselect{\tilde{H}(\sigma,\vec\rho,\omega_0,\omega_1)}}}\rangle\rangle.
\end{align*}
According to (\ref{Hr}) and the condition of the case distinction we have 
$$\rr{\pselect{\tilde{H}(\sigma,\vec\rho,\omega_0,\omega_1)}}{\exists m
  \succeq \omega_0.\phi(\sll{\sigma},m)[\vec s]}$$ or, more detailed, 
$$\rr{\pselect{\tilde{H}(\sigma,\vec\rho,\omega_0,\omega_1)}}{\exists
  m. \We(m) \wedge m \succeq \omega_0 \wedge \phi(\sll{\sigma},m)[\vec
  s]}.$$ From the second conjunct we can conclude, 
$\omega_0 \preceq
\pleft{\pselect{\tilde{H}(\sigma,\vec\rho,\omega_0,\omega_1)}}$. It
remains to show that $\omega_1 \preceq
\pright{\pright{\pselect{\tilde{H}(\sigma,\vec\rho,\omega_0,\omega_1)}}}$.
We have $\rr{\omega_1}{\phi(\sigma,\omega_0)[\vec s]}$ and
$\rr{\pright{\pright{\pselect{\tilde{H}(\sigma,\vec\rho,\omega_0,\omega_1)}}}}{\phi(\sll{\sigma},\pleft{\pselect{\tilde{H}(\sigma,\vec\rho,\omega_0,\omega_1)}})[\vec
s]}$. Now, it is important that $\phi$ is a positive, $\We$-free
formula \emph{without disjunction}.\randnotiz{More detailed IV.8} For these class of formulas, the
realizers do not depend on the terms occurring in them (as long as they
are realizable, of course). Thus, $\omega_1$ and
$\pright{\pright{\pselect{\tilde{H}(\sigma,\vec\rho,\omega_0,\omega_1)}}}$
are equal. Now, the monotonicity follows from the properties we have
for the monotonicity relation together with the
pairing (see Remark~\ref{lepreceq}).

\begin{remark}\label{dis} The proof of the monotonicity property of
  the step function depends on our restriction to disjunction-free
  formulas in the monotone induction scheme. In fact, if we allow
  disjunctions, the monotonicity is not any longer guaranteed, as,
  depending on the terms, different disjuncts could be realized and
  the value of the realizers may differ. In fact, disjunction has a
  ``non-monotonic'' flavor. However, it is not clear whether one can
  make any use of disjunction to enlarge the class of provably
  total functions.  So, we pose as a question:
\begin{question}
What is the class of provably total functions of $\aph$ if the
monotone induction scheme allows disjunctions in the formula $\phi(y,n)$?
\end{question}\rrandnotiz{primitive positive formulae}
\end{remark}

The final result follows now as a corollary:

\begin{corollary}[cf.\ {\cite[Corollary 16]{Str03}}]
Let $t$ be a closed term and assume that 
$$\aphp \vdash \seq{\We(u_1) \wedge \dots \wedge
  \We(u_n)}{\We(t\,u_1\,\dots\,u_n)},$$
for distinct variables $u_1,\dots,u_n$. Then there exists a function
$f: \WW^n \to \WW$ in $\FPH$ such that we have for all words
$w_1,\dots,w_n$ in $\WW$,
$$\mathcal{M}(\lambda\eta) \models
t\,\overline{w_1}\dots\overline{w_n} = \overline{F(w_1,\dots,w_n)}.$$
\end{corollary}

\bibliographystyle{alpha}
\bibliography{compl}

\end{document}